\documentclass{stacs_proc}
\usepackage{amsmath,amssymb,graphics}
\newcommand{\KA}{\textit{KA}\,}
\newcommand{\KS}{C}
\newcommand{\KP}{K}

\let\ge=\geqslant
\let\le=\leqslant

\begin{document}
\title{Limit complexities revisited}
\author[lab1]{L.~Bienvenu}{Laurent Bienvenu}
\address[lab1]{Laboratoire d'Informatique Fondamentale\newline
CNRS \& Universit\'e de Provence,
\newline 39 rue Joliot Curie, F-13453 Marseille cedex 13}
\email{Laurent.Bienvenu@lif.univ-mrs.fr}
\author[lab2]{An.~Muchnik}{Andrej Muchnik}
\address[lab2]{Andrej Muchnik (24.02.1958 -- 18.03.2007)\newline
worked
in the Institute of New Technologies in Education, Moscow}
\author[lab3]{A.~Shen}{Alexander Shen}
\address[lab3]{Laboratoire d'Informatique Fondamentale,\newline
Poncelet Laboratory, CNRS, IITP RAS, Moscow}
\email{Alexander.Shen@lif.univ-mrs.fr}
\author[lab4]{N.~Vereshchagin}{Nikolay Vereshchagin}
\address[lab4]{Moscow State Lomonosov University, Russia}
\email{ver@mccme.ru}

\keywords{Kolmogorov complexity, limit complexities, limit
frequencies, 2-randomness, low basis}


\begin{abstract}
The main goal of this paper is to put some known results in a
common perspective and to simplify their proofs.

We start with a simple proof of a result from~\cite{ver} saying
that $\limsup_{n}\KS(x|n)$ (here $\KS(x|n)$ is conditional
(plain) Kolmogorov complexity of $x$ when $n$ is known) equals
$\KS^{\mathbf{0}'}(x)$, the plain Kolmogorov complexity with
$\mathbf{0}'$-oracle.

Then we use the same argument to prove similar results for
prefix complexity (and also improve results of~\cite{muchnik}
about limit frequencies), a priori probability on binary tree
and measure of effectively open sets. As a by-product, we get a
criterion of $\mathbf{0}'$ Martin-L\"of randomness (called also
$2$-randomness) proved in~\cite{miller}: a sequence $\omega$ is
$2$-random if and only if there exists $c$ such that any prefix
$x$ of $\omega$ is a prefix of some string $y$ such that
$\KS(y)\ge |y|-c$. (In the 1960ies this property was suggested in
\cite{kolmogorov} as one of possible randomness definitions; its
equivalence to $2$-randomness was shown in~\cite{miller} while
proving another $2$-randomness criterion (see also~\cite{nies}):
$\omega$ is $2$-random if and only if $\KS(x)\ge |x|-c$ for some
$c$ and infinitely many prefixes $x$ of $\omega$.

Finally, we show that the low-basis theorem can be used to get
alternative proofs for these results and to improve the result
about effectively open sets; this stronger version implies the
$2$-randomness criterion mentioned in the previous sentence.

\end{abstract}

\maketitle

\stacsheading{2008}{73-84}{Bordeaux}
\firstpageno{73}

\section{Plain complexity}

By $\KS(x)$ we mean the plain complexity of a binary string $x$
(the length of the shortest description of $x$ when an optimal
description method is fixed, see~\cite{li-vitanyi}; no
requirements about prefixes). By $\KS(x|n)$ we mean conditional
complexity of $x$ when $n$ is given~\cite{li-vitanyi}.
Superscript $\mathbf{0}'$ in $\KS^{\mathbf{0}'}$ means that we
consider the relativized (with oracle $\mathbf{0}'$, the
universal enumerable set) version of complexity.

The following result was proved in~\cite{ver}. We provide a
simple proof for it.

\begin{theorem}
        \label{plain}
        $$
\limsup_{n\to\infty}\, \KS(x|n)=\KS^{\mathbf{0}'}(x)+O(1).
        $$
\end{theorem}

\begin{proof}
We start with the easy part. Let $\mathbf{0}_n$ be the
(finite) part of the universal enumerable set that appeared
after $n$ steps. If $\KS^{\mathbf{0}'}(x)\le k$, then there
exists a description (program) of size at most $k$ that
generates $x$ using $\mathbf{0}'$ as an oracle. Only finite part
of the oracle can be used, so $\mathbf{0}'$ can be replaced by
$\mathbf{0}_n$ for all sufficiently large $n$, and oracle
$\mathbf{0}_n$ can be reconstructed if $n$ is given as a
condition. Therefore, $\KS(x|n)\le k+O(1)$ for all sufficiently
large $n$, and
        $$
\limsup_{n\to\infty}\,\KS(x|n)\le \KS^{\mathbf{0}'}(x)+O(1).
        $$

Now fix $k$ and assume that $\limsup\,\KS(x|n) <k$. This means that for all
sufficiently large $n$ the string $x$ belongs to the set
        $$
U_n=\{ u \mid \KS(u|n)<k\}.
        $$
The family $U_n$ is an enumerable family of sets (given $n$ and $k$, we
generate $U_n$); each of these sets has less than $2^k$
elements. We need to construct a $\mathbf{0}'$-computable
process that given $k$ generates at most $2^k$ elements, and among them
all elements that belong to $U_n$ for all sufficiently
large~$n$. (Then strings of length $k$ may be assigned as
$\mathbf{0}'$-computable codes of all generated elements.)

To describe this process, consider the following
operation: for some $u$ and $N$ add $u$ to all $U_n$ such
that $n\ge N$. (In other terms, we add a horizontal ray starting
from $(N,u)$ to the set $\mathcal{U}=\{(n,u)\mid u\in U_n\}$.)
This operation is \emph{acceptable} if all $U_n$ still have less
than $2^k$ elements after it (i.e., if before this operation all
$U_n$ such that $n\ge N$ either contain $u$ or have less than
$2^k-1$ elements).

For given $u$ and $k$ we can find out using $\mathbf{0}'$-oracle
whether this operation is acceptable. Now for all pairs $(N,u)$
(in some computable order) we perform $(N,u)$-operation if it is
acceptable. (The elements added to some $U_i$ remain there
and are taken into account when next operations are attempted.)
This process is $\mathbf{0}'$-computable since after any finite
number of operations the family $\mathcal{U}$ is enumerable
(without any oracle) and its enumeration algorithm can be
$\mathbf{0}'$-effectively found (uniformly in~$k$).

Therefore the set of all elements $u$ that participate in
acceptable operations during this process is uniformly
$\mathbf{0}'$-enumerable. This set contains less than $2^k$
elements (otherwise $U_n$ would become too big for large $n$).
Finally, this set contains all $u$ such that $u$ belongs to the
(initial) $U_n$ for all sufficiently large $n$. Indeed, the
operation is always acceptable if all added elements are
already present.
\end{proof}

The proof has the following structure. We have an enumerable family
of sets $U_n$ that have less than $2^k$ elements. This implies
that the set
        $$
U_\infty=\liminf_{n\to\infty} U_n
        $$
has less than $2^k$ elements (the $\liminf$ of a sequence of sets
is the set of elements that belong to almost all sets of the
sequence). If this set were
$\mathbf{0}'$-enumerable, we would be done. However, this may be
not the case: the criterion
        $$
u\in U_{\infty} \Leftrightarrow \exists N\,(\forall n\ge N)\, [u\in U_n]
        $$
has $\exists\forall$ prefix before an enumerable (not
necessarily decidable) relation, that is, one quantifier more
than we want (to guarantee that $U_\infty$ is
$\mathbf{0}'$-enumerable). However, in our proof we managed to
cover $U_\infty$ by a set that is $\mathbf{0}'$-enumerable and still has
less than $2^k$ elements.

\section{Prefix complexity and a priori probability}

Now we prove similar result for prefix complexity (or, in other
terms, for a priori probability). Let us recall the definition.
The function $a(x)$ on binary strings (or integers) with
non-negative real values is called a \emph{semimeasure} if
$\sum_x a(x) \le 1$. The function~$a$ is \emph{lower
semicomputable} if there exists a computable total function
$(x,n)\mapsto a(x,n)$ with rational values such that for
every~$x$ the sequence $a(x,0), a(x,1),\ldots$ is a
non-decreasing sequence that has limit~$a(x)$.

There exists a maximal (up to a constant factor) lower
semicomputable semimeasure $m$. The value $m(x)$ is sometimes
called the \emph{a priori probability} of $x$. In the same way
we can define \emph{conditional} a priory probability $m(x|n)$ and
$\mathbf{0}'$-\emph{relativized} a priori probability
$m^{\mathbf{0}'}(x)$.

\begin{theorem}
        \label{apriori}
        $$
\liminf_{n\to\infty}\,m(x|n)=m^{\mathbf{0}'}(x)
        $$
up to a $\Theta(1)$ factor.
\end{theorem}

(In other terms, two inequalities with $O(1)$ factors hold.)

\begin{proof}
If $m^{\mathbf{0}'}(x)$ is greater that some
$\varepsilon$, then for some  $k$
the increasing sequence $m^{\mathbf{0}'}(x,k)$
that has
limit $m^{\mathbf{0}'}(x)$ becomes greater than $\varepsilon$.
The
computation of $m^{\mathbf{0}'}(x,k)$
uses only finite amount of information about the
oracle, thus for all sufficiently large  $n$ we have
$m^{\mathbf{0}_n}(x)\ge m^{\mathbf{0}_n}(x,k)>\varepsilon$.
So, similar to the previous theorem, we have
        $$
\liminf_{n\to\infty}\,m(x|n)\ge
\liminf_{n\to\infty}\,m^{\mathbf{0}_n}(x)\ge
m^{\mathbf{0}'}(x)
        $$
up to $O(1)$ factors.

In the other direction the proof is also similar to the previous
one. Instead of enumerable finite sets $U_n$ now we have a
sequence of (uniformly) lower semicomputable functions $x\mapsto
m_n(x)=m(x|n)$. Each of $m_n$ is a semimeasure. We need to
construct a lower $\mathbf{0}'$-semicomputable semimeasure $m'$
such that
        $$
m'(x)\ge \liminf_{n\to\infty}\,m_n(x)
        $$
Again, the $\liminf$ itself cannot be used as $m'$: though
$\sum_x \liminf_n m_n(x)<1$ if $\sum_x m_n(x)\le 1$ for all
$n$, but, unfortunately, the equivalence
        $$
r < \liminf_{n\to\infty} a_n \Leftrightarrow
(\exists r'>r)(\exists N)\, (\forall n\ge N)\, [r'<a_n]
        $$
has too many quantifier alternations (one more than needed; note
that lower semicomputable $a_n$ makes $[\ldots]$ condition
enumerable). The similar trick helps. For a triple $(r,N,u)$
consider an \emph{increase operation} that increases all values
$m_n(u)$ such that $n\ge N$ up to a given rational number $r$
(not changing them if they were greater than or equal to~$r$).
This operation is \emph{acceptable} if all $m_n$ remain
semimeasures after the increase.

The question whether operation is acceptable is
$\mathbf{0}'$-decidable; if it is, we get a new (uniformly)
lower semicomputable (without any oracle) sequence of
semimeasures and can repeat an attempt to perform an increase
operation for some other triple. Doing that for all triples (in
some computable ordering), we can then define $m'(u)$ as the
upper bound of $r$ for all successful $(r,N,u)$ increase
operations (for all $N$). This gives a $\mathbf{0}'$-lower
semicomputable function; it is a semimeasure since we verify the
semimeasure inequality for every successful increase attempt;
finally, $m'(u) \ge \liminf\, m_n(u)$ since if $m_n(u)\ge r$ for
all $n\ge N$, then $(r,N,u)$-increase does not change anything
and is guaranteed to be acceptable.
\end{proof}

The expression $-\log m(x)$ equals the so-called \emph{prefix}
complexity $\KP(x)$ (up to $O(1)$ term; see~\cite{li-vitanyi}).
The same is true for relativized and conditional versions, an we
get the following reformulation of the last theorem:

\begin{theorem}
        \label{prefix}
        $$
\limsup_{n\to\infty}\,\KP(x|n)= \KP^{\mathbf{0}'}(x)+O(1).
        $$
\end{theorem}

Another corollary improves a result of~\cite{muchnik}. For any
(partial) function $f$ from $\mathbb{N}$ to $\mathbb{N}$ we
define the \emph{limit frequency} of an integer $x$ as
        $$
q_f(x)=\liminf_{n\to\infty}\, \frac{\#\{i<n\mid f(i)=x\}}{n}
        $$
In other words, we look at the fraction of $x$-terms in
$f(0),\ldots,f(n-1)$ (undefined values are also listed) and take
$\liminf$ of these frequencies. It is easy to see that for a
total computable $f$ the function $q_f$ is a lower
$\mathbf{0}'$-semicomputable semimeasure. The argument above
proves the following result:

\begin{theorem}
        \label{frequency}
For any partial computable $f$ the function $q_f$ is upper
bounded by a lower $\mathbf{0}'$-semicomputable semimeasure.
\end{theorem}

In~\cite{muchnik} it is shown that for some total computable $f$
the function $q_f$ is a maximal lower
$\mathbf{0}'$-semicomputable semimeasure and therefore
$\mathbf{0}'$-relativized a priori probability can be defined as
maximal limit frequency for total computable functions. Now we
see that the same is true for partial computable functions:
allowing them to be partial does not increase the maximal limit
frequency.

The similar argument also is applicable to the so-called \emph{a
priori complexity} defined as negative logarithm of a maximal
lower semicomputable semimeasure on the binary tree
(see~\cite{zvonkin-levin}). This complexity is sometimes denoted
as $\KA(x)$ and we get the following statement:

\begin{theorem}
        \label{aprioritree}
        $$
\limsup_{n\to\infty}\KA(x|n)=\KA^{\mathbf{0}'}(x)+O(1).
        $$
\end{theorem}

(To prove this we define an increase operation in such a way
that it increases not only $a(x)$ but also $a(y)$ for $y$ that
are prefixes of $x$, if necessary. The increase is acceptable if
$a(\Lambda)$ still does not exceed~$1$.)

It would be interesting to find out whether similar results are
true for monotone complexity or not (the authors do not know this).

\section{Open sets of small measure}

We now try to apply the same trick in a slightly different
situation, for effectively open sets. The Cantor space $\Omega$
is a set of all infinite sequence of zeros and ones. An
\emph{interval} $\Omega_x$ (for a binary string $x$) is formed
by all sequences that have prefix $x$. Open sets are unions of
intervals. An \emph{effectively open} subset of $\Omega$ is an
enumerable union of intervals, i.e., the union of intervals
$\Omega_x$ where $x$ are takes from some enumerable set of
strings.

We consider standard (uniform Bernoulli) measure on $\Omega$:
the interval $\Omega_x$ has measure $2^{-l}$ where $l$ is the
length of $x$.

A classical theorem of measure theory says: \emph{if
$U_0,U_1,U_2,\ldots$ are open sets of measure at most
$\varepsilon$, then $\liminf_n U_n$ has measure at most
$\varepsilon$, and this implies that for every
$\varepsilon'>\varepsilon$ there exists an open set of measure
at most $\varepsilon'$ that covers $\liminf_n U_n$.}

Indeed,
        $$
\liminf_{n\to\infty}\,U_n = \bigcup_{N} \bigcap_{n\ge N} U_n,
        $$
and the measure of the union of an increasing sequence
        $$
V_N = \bigcap_{n\ge N} U_n,
        $$
equals the limit of measures of $V_N$, and all these measures do
not exceed $\varepsilon$ since $V_N\subset U_N$. It remains to
note that for any measurable set $X$ its measure is the infimum
of the measures of open sets that cover~$X$.

We now can try to ``effectivize'' this statement in the same way
as we did before. First we started with an (evident) statement:
\emph{if $U_n$ are finite sets of at most $2^k$ elements, then
$\liminf_n U_n$ has at most $2^k$ elements} and proved its
effective version: \emph{for a uniformly enumerable family of
open sets $U_n$ that have at most $2^k$ elements, the set
$\liminf_n U_n$ is contained in a uniformly
$\mathbf{0}'$-enumerable set that has at most $2^k$ elements.}
Then we did similar thing with semimeasures (again, the
non-effective version is trivial: it says that if $\sum_x
m_n(x)\le 1$ for every $n$, then $\sum_x\liminf_n m_n(x)\le
1$).

Now the effective version could look like this. \emph{Let
$\varepsilon>0$ be a rational number and let $U_0,U_1,\ldots$ be an
enumerable family of effectively open sets of measure at most
$\varepsilon$ each. Then for every rational
$\varepsilon'>\varepsilon$ there exists a
$\mathbf{0}'$-effectively open set of measure at most
$\varepsilon'$ that contains} $\liminf_{n\to\infty} U_i=
\bigcup_{N} \bigcap_{n\ge N} U_n$.

However, the authors do not know whether this is always true. The
argument that we have used can nevertheless be applied do prove
the following weaker version:

\begin{theorem}
        \label{opensetinterior}
Let $\varepsilon>0$ be a rational number and let $U_n$ be an
enumerable family of effectively open sets of measure at most
$\varepsilon$ each. Then there exists a uniformly
$\mathbf{0}'$-effectively open set of measure at most
$\varepsilon$ that contains
        $$
\bigcup_{N} \mathrm{Int}\bigl(\bigcap_{n\ge N} U_n\bigr)
        $$
\end{theorem}

Here $\mathrm{Int}(X)$ denotes the interior part of $X$, i.e.,
the union of all open subsets of $X$. In this case we do not
need $\varepsilon'$ (which one could expect since the union of
open sets is open).

\begin{proof}
Following the same scheme, for every string $x$ and
integer $N$ we consider $(x,N)$-operation that adds $\Omega_x$
to all $U_n$ such that $n\ge N$. This operation is
\emph{acceptable} if measures of all $U_n$ remain at most
$\varepsilon$ for each $n$. This can be checked using
$\textbf{0}'$-oracle (if the operation is not acceptable, it
becomes known after a finite number of steps).

We attempt to perform this operation (if acceptable) for all
pairs in some computable order. The union of all added intervals
for all accepted pairs is $\mathbf{0}'$-effectively open. If
some sequence belongs to the union of the interior parts, then it is
covered by some interval $\Omega_u$ that is a subset of $U_n$
for all sufficiently large $n$. Then some $(u,N)$-operation is
acceptable since it actually does not  change anything and
therefore $\Omega_u$ is a part of an $\mathbf{0}'$-open set that
we have constructed.
\end{proof}

\section{Kolmogorov and $2$-randomness}

This result has an historically remarkable corollary. When
Kolmogorov tried to define randomness in 1960ies, he started
with the following approach. A sequence $x$ of length $n$ is
``random'' if its complexity $\KS(x)$ (or conditional complexity
$\KS(x|n)$; in fact, these requirements are almost equivalent)
is close to $n$: the \emph{randomness deficiency} $d(x)$ is
defined as the difference $|x|-\KS(x)$ (here $|x|$ stands for
the length of~$x$). This sounds reasonable, but if we then
define a random sequence as a sequence whose prefixes have
deficiencies bounded by a constant, such a sequence does not exist at all:
Martin-L\"of showed that every infinite sequence has prefixes
of arbitrarily large deficiency, and suggested a different
definition of randomness using effectively null sets. Later more
refined versions of randomness deficiency (using monotone or
prefix complexity) appeared that make the criterion of
randomness in terms of deficiencies possible. But before that,
in 1968, Kolmogorov wrote: ``The most natural definition of
infinite Bernoulli sequence is the following: $x$ is considered
$m$-Bernoulli type if $m$ is such that all $x^i$ are
\emph{initial segments} of the finite $m$-Bernoulli sequences.
Martin-L\"of gives another, possibly narrower
definition'' (\cite{kolmogorov}, p.~663).

Here Kolmogorov speaks about ``$m$-Bernoulli'' finite sequence $x$
(this means that $\KS(x|n,k)$ is greater than
$\log\binom{n}{k}-m$ where $n$ is the length of $x$ and $k$ is
the number of ones in $x$). For the case of uniform Bernoulli
measure (where $p=q=1/2$) one would reformulate this definition
as follows. Let us define
        $$
\bar d (x)=\inf \{d(y)\mid \text{$x$ is a prefix of $y$}\}
        $$
and require that $\bar d(x)$ is bounded for all prefixes of an
infinite sequence $\omega$. It is shown by J.~Miller in \cite{miller}
that this definition is equivalent to Martin-L\"of randomness
relativized to $\mathbf{0}'$ (called also $2$-\emph{randomness}):

\begin{theorem}
        \label{kolmogorovrandomness}
A sequence $\omega$ is Martin-L\"of $\mathbf{0}'$-random if and
only if the quantities $\bar d(x)$ for all prefixes $x$ of $\omega$
are bounded by a \textup(common\textup) constant.
\end{theorem}

In turns out that this result (in one direction) easily follows
from the previous theorem.

\begin{proof}
Assume that $\bar d$-deficiencies for prefixes of
$\omega$ are not bounded. According to Martin-L\"of definition,
we have to construct for a given $c$ an $\mathbf{0}'$-effectively
open set
that covers $\omega$ and has measure at most $2^{-c}$.

Fix some $c$. For each $n$ consider the set $D_n$ of all
sequences $u$ of length $n$ such that $\KS(u)<n-c$ (i.e.,
sequences $u$ of length~$n$ such that $d(u)>c$). It has at most
$2^{n-c}$ elements. The requirement $\bar d(x)>c$ means that
every string extension $y$ of $x$ belongs to $D_m$ where $m$ is
its length. This implies that $\Omega_x$ is contained in every
$U_m$ where $m\ge |x|$ and $U_m$ is the set of all sequences
that have prefixes in $D_m$ (this set has measure at most
$2^{-c}$). Therefore, in this case the interval $\Omega_x$ is a
subset of $\bigcap_{m\ge |x|} U_m$ and (being open) is a subset of
its interior. Then we conclude (using the result proved above)
that $\Omega_x$ (=every sequence with prefix $x$) is covered by an
$\mathbf{0}'$-effectively open set of measure at most $2^{-c}$
constructed as explained above. So if some $\omega$ has prefixes
of arbitrarily large $\bar d$-deficiency, then $\omega$ is not
$\mathbf{0}'$ Martin-L\"of random.

Note that this argument works also for conditional complexity
(with length as condition) and gives a slightly stronger result.

For the sake of completeness we reproduce (from~\cite{miller})
the proof of the reverse implication (essentially unchanged).
Assume that a sequence $\omega$ is covered (for each $c$) by a
$\mathbf{0}'$-computable sequence of intervals $I_0,I_1,\ldots$
of total measure at most $2^{-c}$. (We omit $c$ in our notation,
but all these constructions depend on~$c$.)

Using the approximations $\mathbf{0}_n$ instead of full
$\mathbf{0}'$ and performing at most $n$ steps of computation
for each $n$ we get another (now computable) family of intervals
$I_{n,0},I_{n,1},\ldots$ such that $I_{n,i}=I_i$ for every $i$
and sufficiently large $n$. We may assume without loss of
generality that $I_{n,i}$ either has size at least $2^{-n}$
(i.e., is determined by a string of length at most~$n$) or
equals $\bot$ (a special value that denotes the empty set) since
only the limit behavior is prescribed. Moreover, we may also
assume that $I_{n,i}=\bot$ for $i>n$ and that the total measure
of all $I_{n,0},I_{n,1},\ldots$ does not exceed $2^{-c}$ for
every $n$ (by deleting the excessive intervals in this order;
the stabilization guarantees that all limit intervals will be
eventually let through).

Since $I_{n,i}$ is defined by intervals of size at least
$2^{-n}$, we get at most $2^{n-c}$ strings of length $n$ covered
by intervals $I_{n,i}$ for given $n$ and all $i$. This set is
decidable (recall that only $i$ not exceeding $n$ are used),
therefore each string in this set can be defined (assuming $c$
is known) by a string of length $n-c$, binary representation of
its ordinal number in this set. (Note that this string also
determines~$n$ if $c$ is known.)

Returning to the sequence~$\omega$, we note that it is covered by
some $I_i$ and therefore is covered by $I_{n,i}$ for this $i$
and all sufficiently large $n$ (after the value is stabilized),
say, for all $n\ge N$. Let $u$ be a prefix of $\omega$ of length
$N$. All continuations of $u$ of any length $n$ are covered by
$I_{n,i}$ and have complexity less than $n-c+O(1)$. In fact,
this is a conditional complexity with condition $c$; we get
$n-c+2\log c+O(1)$, so $\bar d(u)\ge c - 2 \log c -O(1)$.

Such a string $u$ can be found for every $c$, therefore $\omega$
has prefixes of arbitrarily large $\bar d$-deficiency.
\end{proof}

In fact a stronger statement than
Theorem~\ref{kolmogorovrandomness} is proved
in~\cite{miller,nies}; our tools are still too weak to get this
statement. However, the low basis theorem helps.

\section{The low basis theorem}

This is a classical result in recursion theory (see,
e.g.,~\cite{lowbasis}). It was used in~\cite{nies} to prove
$2$-randomness criterion; analyzing this proof, we get theorems
about limit complexities as byproducts. For the sake of
completeness we reproduce the statement and the proof of
low-basis theorem here; they are quite simple.

\begin{theorem}
        \label{lowbasis}
Let $U\subset\Omega$ be an effectively open set that does not
coincide with $\Omega$. Then there exists a sequence
$\omega\notin U$ which is low, i.e., $\omega'=\mathbf{0}'$
\end{theorem}

Here $\omega'$ is the jump of $\omega$; the equation
$\omega'=\mathbf{0}'$ means that the universal $\omega$-enumerable
set is $\mathbf{0}'$-decidable.

Theorem~\ref{lowbasis}
says that any effectively closed non-empty set contains
a low element. For example, if $P,Q\subset\mathbb{N}$ are
enumerable inseparable sets, then the set of all separating
sequences is an effectively closed set that does not contain
computable sequences. We conclude, therefore, that there exists a
non-computable low separating sequence.

\begin{proof}
Assume that an oracle machine $M$ and an input $x$ are
fixed. The computation of $M$ with oracle $\omega$ on $x$ may
terminate or not depending on oracle $\omega$. Let us consider
the set $T(M,x)$ of all $\omega$ such that $M^\omega(x)$
terminates (for fixed machine $M$ and input $x$). This set is an
effectively open set (if termination happens, it happens due to
finitely many oracle values). This set together with $U$ may
cover the entire~$\Omega$; this means that $M^\omega(x)$
\emph{terminates for all $\omega\notin U$}. If it is not the
case, we can add $T(M,x)$ to $U$ and get a bigger effectively
open set $U'$ that still has non-empty complement such that
$M^\omega(x)$ \emph{does not terminate for all $\omega\in U'$}.
This operation guarantees (in one of two ways) that termination
of the computation $M^\omega(x)$ does not depend on the choice
of $\omega$ (in the remaining non-empty effectively closed set).

This operation can be performed for all pairs $(M,x)$
sequentially. Note that if $U\cup T(M,x)$ covers the entire $\Omega$,
this happens on some finite stage (compactness), so
$\mathbf{0}'$ is enough to find out whether it happens or not,
and on the next step we have again some effectively open
(without any oracle) set. So $\mathbf{0}'$-oracle is enough to
say which of the computations $M^\omega(x)$ terminate (as we
have said, this does not depend of the choice of $\omega$).
Therefore any such $\omega$ is low (the universal
$\omega$-enumerable set is $\mathbf{0}'$-decidable). And such an
$\omega$ exists since the intersection of the decreasing
sequence of non-empty closed sets is non-empty (compactness).
\end{proof}

\section{Using the low basis theorem}

Let us show how Theorem~\ref{plain} can be proved using the low
basis theorem. As we have seen, we have an enumerable family of
sets $U_n$ that have at most $2^k$ elements and need to
construct effectively a $\mathbf{0}'$-enumerable set that has at
most $2^k$ elements and contains $U_\infty=\liminf_n U_n$.

If the sets $U_n$ are (uniformly) decidable, then $U_\infty$ is
$\mathbf{0}'$-enumerable and we do not need any other set. The low
basis theorem allows us to reduce general case to this special
one. Let us consider the family of all ``upper bounds'' for
$U_n$: by an upper bound we mean a sequence $V_n$ of finite sets
that contain $U_n$ and still have at most $2^k$ elements each.
The sequence $V_0,V_1,\ldots$ can be encoded as an infinite binary
sequence (first we encode $V_0$, then $V_1$ etc.; note that each
$V_i$ can be encoded by a finite number of bits though this
number depends on $V_i$).

For a binary sequence the property ``to be an encoding of an
upper bound for $U_n$'' is effectively closed (the restriction
$\#V_n < 2^k$ is decidable and the restriction $U_n \subset V_n$
is co-enumerable). Therefore the low basis theorem can be applied.
We get an upper bound $V$ that is low. Then $V_\infty=\liminf
V_n$ is (uniformly in $k$) $V'$-enumerable (as we have said:
with $V$-oracle the family $V_n$ is uniformly decidable), but
since $V$ is low, $V'$-oracle can be replaced by
$\mathbf{0}'$-oracle, and we get the desired result.

This proof though being simple looks rather mysterious: we get
something almost out of nothing! (As far as we know, this idea
in a more advanced context appeared in~\cite{nies}.)

The same trick can be used to prove Theorem~\ref{apriori}: here
``upper bounds'' are distributions $M_n$ with rational values
and finite support that are greater than $m(x|n)$ but still are
semimeasures. (Technical correction: first we have to assume
that $m(x|n)=0$ if $x$ is large, and then we have to weaken the
restriction $\sum M_n(x)\le 1$ replacing $1$ by, say, $2$; this
is needed since the values $m(x|n)$ may be irrational.)

Theorem~\ref{aprioritree} can be also proved in this way (upper
bounds should be semimeasures on tree with rational values and
finite support).

As to Theorem~\ref{opensetinterior}, here the application of the
low basis theorem allows us to get a stronger result than before
(though not the most strong version we mentioned as an open question):

\begin{theorem}
        \label{opensetinteriorstrong}
Let $\varepsilon>0$ be a rational number and let $U_n$ be an
uniformly enumerable family of effectively open sets, i.e.,
        $$
U_n=\cup \{\Omega_x \mid (n,x)\in U\}
        $$
for some enumerable set $U\subset\mathbb{N}\times\{0,1\}^*$.
Assume that $U_n$ has measure at most $\varepsilon$ for
every~$n$. Assume also that $U_i$ has ``effectively bounded
granularity'', i.e., all strings $x$ such that $(n,x)\in U$ have
length at most $c(n)$ where $c$ is a total computable function.
Then for every $\varepsilon'>\varepsilon$ there exists a
$\mathbf{0}'$-effectively open set $W$ of measure at most
$\varepsilon'$ that contains
        $$
\liminf_{n\to\infty} U_n=\bigcup_{N}\bigcap_{n\ge N} U_n
        $$
and this construction is uniform.
\end{theorem}

\begin{proof}
First we use the low basis theorem to reduce the general
case to the case where $U$ is decidable and for every $(n,x)\in U$
the length of $x$ is exactly $c(n)$.

Indeed, define an ``upper bound'' as a sequence $V$ of sets
$V_n$ where $V_n$ is a set of strings of length $c(n)$ such that
$U_n$ is covered by the intervals generated by elements of
$V_n$. Again $V$ can be encoded as an infinite sequence of zeros
and ones, and the property ``to be an upper bound'' is
effectively closed. Applying the low basis theorem, we choose a low
$V$ and add it is an oracle. Since $V'$ is equivalent to $\mathbf{0}'$,
for our purpose we may assume that $V$ is decidable.

Now we have to deal with the decidable case. Let us represent
the set $U_\infty$ as a union of the disjoint sets
        $$
F_0=\bigcap_i U_i,\
F_1=\bigcap_{i\ge 1} U_i \setminus U_0,\
F_2=\bigcap_{i\ge 2} U_i \setminus U_1,\ldots
        $$
(for each element $x$ in $U_\infty$ we consider the last $U_i$
that does not contain~$x$). Each of $F_i$ is (in the decidable
case) an effectively closed set (recall than $U_i$ is
open-closed due to the restriction on $c(i)$).
Moreover, the $F_i$ are pairwise disjoint and the family $F_i$
satisfies
        $$
\liminf_{n\to+\infty}\, U_n = \bigcup_i F_i
        $$
and thus
        $$
\sum_i \mu(F_i)=\mu(\liminf_{n\to+\infty}\, U_n).
        $$
The measure of
each of $F_i$ is $\mathbf{0}'$-computable, and using
$\mathbf{0}'$-oracle we can find a finite set of intervals that
covers $F_i$ and has measure
        $$
\mu(F_i)+(\varepsilon'-\varepsilon)/2^{i+1}
        $$
Putting
all these intervals together, we get the desired set $W$. So the
decidable case (and therefore the general one, thanks to low
basis theorem) is completed.
\end{proof}

\section{Corollary on 2-randomness}

Theorem~\ref{opensetinteriorstrong} can be used to prove
$2$-randomness criterion from~\cite{miller,nies}. In fact, this
gives exactly the proof from~\cite{nies}; the only thing we did
is structuring the proof in two parts (formulating
Theorem~\ref{opensetinteriorstrong} explicitly and putting it in
the context of other results on limits of complexities).

\begin{theorem}[\cite{miller,nies}]
        \label{miller-nies}
A sequence $\omega$ is $\mathbf{0}'$ Martin-L\"of random if and only if
        $$
\KS(\omega_0\omega_1\ldots\omega_{n-1})\ge n - c
        $$
for some $c$ and for infinitely many~$n$.
\end{theorem}

\begin{proof}
Let us first understand the relation between this
theorem and Theorem~\ref{kolmogorovrandomness}. If
        $$
\KS(\omega_0\omega_1\ldots\omega_{n-1})\ge n - c
        $$
for infinitely many $n$ and given $c$, then $\bar d(x)\le c$
for every prefix $x$ of $\omega$ (indeed, one can find the
required continuation of $x$ among prefixes of $\omega$). As we
know, this guarantees that $\omega$ is $\mathbf{0}'$
Martin-L\"of random.

It remains to prove that if for all $c$ we have
        $$
\KS(\omega_0\omega_1\ldots\omega_{n-1})< n - c
        $$
for all sufficiently large $n$, then $\omega$ is not
$\mathbf{0}'$-random. Using the same notation as in the proof of
Theorem~\ref{kolmogorovrandomness}, we can say that $\omega$ has
a prefix in $D_n$ and therefore belongs to $U_n$ for all
sufficiently large $n$. We can apply then
Theorem~\ref{opensetinteriorstrong} since $U_n$ is defined using
strings of length $n$ (so $c(n)=n$) and cover $U_\infty$ (and
therefore $\omega$) by a $\mathbf{0}'$-effectively open set of
small measure. Since this can be uniformly done for all~$c$, the
sequence $\omega$ is not $\mathbf{0}'$-random.
\end{proof}

\textbf{Remark}. The results above may be considered as special
cases of an effective version of a classical theorem in measure
theory: Fatou's lemma. This lemma guarantees that if $\int
f_n(x)\,d\mu(x)\le\varepsilon$ for $\mu$-measurable functions
$f_0,f_1,f_2,\ldots$, then
        $$
\int \liminf_{n\to+\infty} f_n (x)\,d\mu(x) \le \varepsilon.
        $$
The constructive version assumes that $f_i$ are lower
semicomputable and satisfy some additional conditions; it says
that for every $\varepsilon'>\varepsilon$ there exists a lower
$\mathbf{0}'$-semicomputable function $\varphi$ such that
$\liminf\,f_n (x)\le \varphi(x)$ for every $x$ and
$\int\varphi(x) d\mu(x)\le\varepsilon'$.

\newpage
\null

\end{document}